\documentclass[onecolumn]{article}

\usepackage{amsmath}
\usepackage{amsthm}
\usepackage{physics}
\usepackage[colorlinks=true, allcolors=blue]{hyperref}

\usepackage[letterpaper,top=2cm,bottom=2cm,left=3cm,right=3cm,marginparwidth=1.75cm]{geometry}
\usepackage[utf8]{inputenc}

\usepackage{natbib} 
\usepackage{algorithm}  
\usepackage{algorithmic}  
\usepackage{caption}
\usepackage{subcaption}
\usepackage{multirow}
\usepackage{graphicx}
\usepackage{appendix}
\usepackage{verbatim}
\usepackage{amsfonts}
\usepackage{authblk}
\usepackage{setspace}
\usepackage{bm}

\newcommand{\vV}{\textbf{V}}
\newcommand{\vP}{\textbf{P}}
\newtheorem{prop}{Proposition}
\newtheorem{defi}{Definition}

\usepackage{multirow}
\usepackage{graphicx}
\usepackage{amsfonts}
\fontsize{11}{11}\selectfont

\usepackage[switch]{lineno}

\begin{document}


\title{A prediction-based forward-looking vehicle dispatching strategy for dynamic ride-pooling}

\author{Xiaolei~Wang, Chen~Yang, Yuzhen~Feng, Luohan~Hu, Zhengbing~He, {\it Senior Member, IEEE}
\thanks{The work described in this study was supported by grants from the National Natural Science Foundation of China under Project No. 72022013, No. 72021002, and No.72061127003. Xiaolei Wang was also supported by the Fundamental Research Funds for the Central Universities, and CCF-DiDi GAIA Collaborative Research Funds for Young Scholars. (Corresponding author: Xiaolei Wang)}
\thanks{Xiaolei Wang, Chen Yang, ,Yuzhen Feng and Luohan Hu are with the School of Economics and Management, Tongji University, Shanghai, 200092, China (e-mail: yangchen\_tju@tongji.edu.cn; xiaoleiwang@tongji.edu.cn; fyz020301@outlook.com; lhhu2000@outlook.com).
Zhengbing He is with the Senseable City Lab, Massachusettes Instittue of Technology, Cambridge, MA, United States (e-mail: he.zb@hotmail.com).
}
}

\maketitle
\begin{abstract}
For on-demand dynamic ride-pooling services, e.g., Uber Pool and Didi Pinche, a well-designed vehicle dispatching strategy is crucial for platform profitability and passenger experience. Most existing dispatching strategies overlook incoming pairing opportunities, therefore suffer from short-sighted limitations. In this paper, we propose a forward-looking vehicle dispatching strategy, which first predicts the expected distance saving that could be brought about by future orders and then solves a bipartite matching problem based on the prediction to match passengers with partially occupied or vacant vehicles or keep passengers waiting for next rounds of matching. 
To demonstrate the performance of the proposed strategy, a number of simulation experiments and comparisons are conducted based on the real-world road network and historical trip data from Haikou, China. 
Results show that the proposed strategy outperform the baseline strategies by generating approximately 31\% more distance saving and 18\% less average passenger detour distance. 
It indicates the significant benefits of considering future pairing opportunities in dispatching, and highlights the effectiveness of our innovative forward-looking vehicle dispatching strategy in improving system efficiency and user experience for dynamic ride-pooling services.

\end{abstract}

\textbf{Keyworads:}
Ride-pooling, vehicle dispatching, forward-looking, distance saving, prediction

\section{Introduction}
Dynamic ride-pooling services, such as Uber Pool and Didi Pinche, are on-demand mobility services that enable passengers to share vehicle space and split cost with other passengers during the whole or part of their trips \cite{dimitrakopoulos2011intelligent}. 
To respond to on-demand trip requests, a ride-pooling platform dispatches vacant or partially occupied vehicles to pick up passengers in real-time, and keeps searching for matching orders on the fly. 
As demonstrated in \cite{santi2014quantifying} based on 170 million taxi orders in New York, almost 100\% of taxi trips in New York can be shared with other trips in less than a 2-minute delay, if all taxi passengers were willing to take dynamic ride-pooling services; and the taxi industry in New York can also maintain its instantaneous service quality with 70\% of the instantaneous taxi fleet size \cite{vazifeh2018addressing}.
Similar observations are reported for many other cities and countries such as Munich, Germany, Haikou, China, San Francisco, the United States, and Singapore \cite{tachet2017scaling,engelhardt2019quantifying,zhu2022potential,zwick2021agent,zwick_ride-pooling_2021}.

To provide an attractive and profitable service, an efficient vehicle dispatching strategy is of ultimate importance for a ride-pooling platform. 
In recent years, the design of vehicle dispatching strategies for dynamic ride-pooling has drawn a lot of attention and witnessed significant advancements. In most literature, it is formulated into a dial-a-ride problem (DRP) and the DRP is then solved in a rolling horizon framework \citep{cordeau2007dial,santos2015taxi}. While exact solution methods for the DRP are proposed in \cite{cordeau2007dial,santos2015taxi}, most studies employ the insertion heuristic method to support real-time dispatching decisions \citep{2013Dynamic,ma2013t,hosni2014shared,jung2016dynamic,zwick2022shifts}. After receiving an order, these dispatching strategies generally try different insertions of the order into the job list of nearby vehicles, evaluating the effects of each insertion (on system distance saving, platform's profit, and detours of passengers on-board), and then select the vehicle that mostly aligns with the platform's goals without making on-board passengers worse off. In 2017,  a different on-demand high-capacity vehicle dispatching framework was proposed by \cite{doi:10.1073/pnas.1611675114}. They decouple the problem by first computing feasible trips from a pairwise shareability graph and then assigning trips to vehicles.This framework is general and can be used for many real-time multivechile, multitask assignment problems, and is now widely considered as the state-of-the-art method.   
Nevertheless, all the above dispatching strategies are myopic in the sense that they only maximize instantaneous reward of each dispatching, without taking future pairing opportunities into consideration. 

How to take future pairing opportunities into consideration is the most challenging part in the vehicle dispatching problem for dynamic ridepooling service. 
Consider the simple example as shown in Figure \ref{figure: An example of vehicle dispatching problem}, a waiting passenger $P_1$ can be assigned to one of the partially occupied vehicles $\text{V}_1$ and $\text{V}_2$ and the vacant vehicle $\text{V}_3$, or the platform can keep $P_1$ waiting. Most (if not all) myopic dispatching strategies would assign $P_1$ to the partially occupied vehicles $\text{V}_1$ or $\text{V}_2$, as they bring about immediately rewards (e.g., distance saving). 
However, this does not necessarily imply that dispatching $\text{V}_3$ to $\text{P}_1$ or keeping $\text{P}_1$ waiting is a worse choice. In both cases, P1 may encounter better ride-pooling partners and form a ride-pooling trip with longer distance saving and/or shorter detour distance afterward. 

\begin{figure}[!t]
\centering
\includegraphics[width=3.2in]{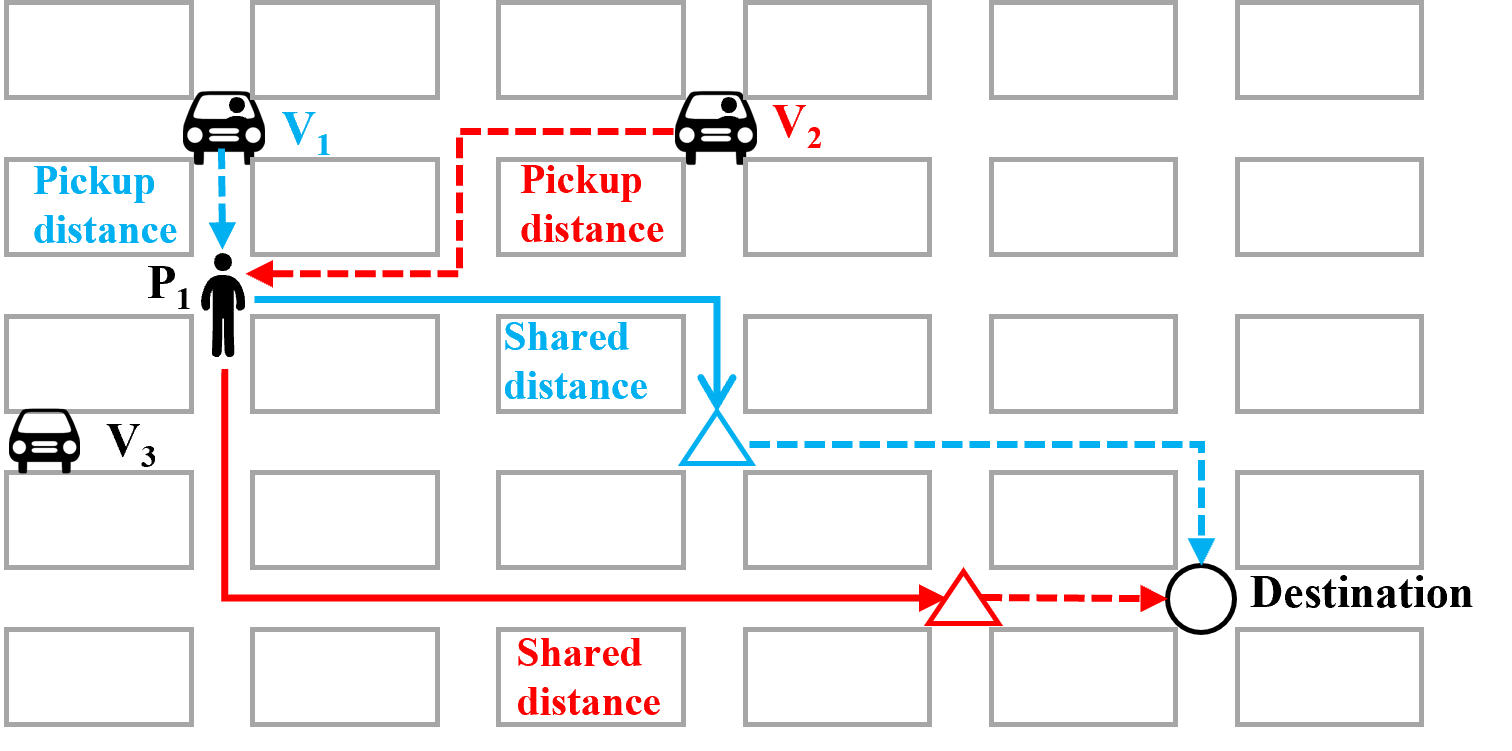}
\caption{An example of vehicle dispatching problem for dynamic ride-pooling ($\text{V}_1$ and $\text{V}_2$ are partially occupied vehicles that can form ridepooling trips with the waiting passenger $\text{P}_1$, and $\text{V}_3$ is a vacant vehicle). The blue and red triangles indicate the destinations of the passengers on the partially occupied vehicles $\text{V}_1$ and $\text{V}_2$, respectively, and the black cycle is the destination of passenger $\text{P}_1$.}
\label{figure: An example of vehicle dispatching problem}
\end{figure}


In recent years, there have also been some studies that focus on forward-looking vehicle dispatching strategy design \citep{xian2020, 2020Neural}. They both formulate the problem as a multi-stage stochastic programming problem, where the platform solves a bipartite matching problem in each stage, and the resulting dispatching decisions together with the random appearance of ride-pooling orders lead to different post-decision states. The utility of each passenger-vehicle matching pair is defined as the difference between the pre-decision state value and the post-decision state value. Specifically, \cite{xian2020} approximated the state value as a linear function of vehicle state variables, and updated the parameters in the approximated value function with the dual values of each matching problem. And \cite{2020Neural} approximated the state value function with a neural network, and updated the network parameters with the gradient descent method. Both methods require a large number of simulations to learn the state values, and the convergence of the state value function is not always guaranteed. Moreover, the learned state values are often difficult to interpret and not resilient to changes in the environment. 

Recently, \cite{wang2021predicting} developed an accurate method to predict the pairing probability, expected ride distance and expected shared distance for each ridepooling order. Assuming that ridepooling orders between each origin and destination (OD) pair follow a Poisson process with a given mean arrival rate, they proposed a system of nonlinear equations to capture the complex matching and competing relationship among orders between different OD pairs in a general road network. Solving the system of nonlinear equations simultaneously yields the pairing probability, expected ride distance and expected shared distance of all OD pairs. This prediction method shows surprisingly satisfactory accuracy in simulation experiments, thus opens up a new path for the design of forward-looking vehicle dispatching strategies.

Enlightened by \cite{wang2021predicting}, this paper presents a new method for the design of forward-looking vehicle dispatching for dynamic ridepooling service. It adopts an offline-prediction and online-matching approach: 1) based on historical data, we first take advantage of the prediction method proposed by \cite{wang2021predicting} to estimate the expected distance saving for passengers between different OD pairs if they were kept waiting or assigned to vacant vehicles; 2) based on the prediction, we set the utility of each match in the bipartite matching problem to generate real-time matching results.  
The main merits of the proposed matching strategy are summarized as follows.
\begin{itemize}
\item The proposed strategy sets the utility of each match according to the predicted distance saving that could be brought about by future orders, therefore is endowed with the forward-looking property. Based on the predicted expected distance saving of different passengers under different options, our strategy deliberately gives up some appeared pairing opportunities, lets some passengers to wait a little bit longer, and assigns some passengers to vacant vehicles immediately. %
  
  \item The proposed matching strategy shows significant improvements over myopic strategies in terms of both total distance saving and average passenger detour distance. As shown by extensive simulation experiments, in comparison with the state-of-the-art strategy\cite{doi:10.1073/pnas.1611675114}, our proposed strategy can generate up to 31\% more total distance saving, and reduce the average passenger detour distance by 18\%. The significant improvements in both system performance and user experience strongly support the value of the predicted information in enhancing the vehicle dispatching decisions. 
  \item The designing method of our proposed forward-looking dispatching strategy is more efficient than previous methods. We obtain the predicted distance saving of each option for each OD pair by solving a system of nonlinear equations based on the historical trip demand. In comparison with existing methods in \cite{xian2020} and \cite{2020Neural}, it avoids the need for time-consuming state-value estimation or extensive training and learning with large-scale dataset. 
\end{itemize}

The rest of this work is organized as follows. 
Section \ref{sec:preliminaries} introduces a preliminary description of the vehicle dispatching problem for dynamic ride-pooling, and describes a myopic matching strategy that serves as a benchmark. 
Section \ref{sec:forward-looking matching strategy} introduces how to determine the utility of each match based on predictions, and then proposes a forward-looking matching strategy. 
Section \ref{sec: numerical study} demonstrates the effectiveness of the proposed strategy through various simulation experiments. 
Finally, Section \ref{sec: conclusion} concludes the paper.

\section{Preliminaries}\label{sec:preliminaries}
\subsection{Problem description}\label{subsec:problem description} 

Consider a ride-pooling platform that conducts a vehicle-passenger batch matching every $\Delta t$ time interval. 
Each waiting passenger can be assigned to a vacant or a partially occupied vehicle, or be delayed to the next round of matching. 
Passengers will cancel their orders if they are not matched within their maximum waiting time, which is denoted by $K$ (seconds). 
In each batch matching, let $\vP$ be the set of passengers waiting for responses, $\vV_0$ be the set of vacant vehicles, and $\vV_1$ be the set of partially occupied vehicles. 
And let $\left\{0\right\}$ be the dummy vehicle representing the option of keeping passengers waiting at the origin for better pairing opportunities in the future. 
To ensure stable ride-pooling service quality, this study assumes a passenger can be matched to a vacant vehicle only if the pick-up distance is less than $\bar{R}$; and a passenger can be matched to a partially occupied vehicle only if the following pairing conditions are met: 
1) the pick-up distance is less than $\bar{R}$; 
2) the resulting detour distance of both passengers does not exceed $\bar{D}$. 
For each passenger $p\in \vP$, let $\vV_{0,p}$ and $\vV_{1,p}$ respectively be the sets of vacant vehicles and partially occupied vehicles that satisfy the above pairing conditions, and $\vV_p=\vV_{0,p}\cup\vV_{1,p}\cup\left\{0\right\}$ be the set of matching options. 
And conversely, for each vehicle $v\in \vV_0\cup\vV_1$, the set of waiting passengers $\vP_v$ can be predetermined based on the aforementioned pairing conditions. 
Let $u(v,p)$ be the utility of assigning a passenger $p\in \vP$ to a vehicle $v\in \vV_p=\vV_{0,p}\cup\vV_{1,p}\cup\left\{0\right\}$, and $\boldsymbol{x}=\left\{x_{v,p},v\in \vV_p,p\in \vP\right\}$ be the vector of binary decision variables with $x_{v,p}=1$ if passenger $p\in \vP$ is assigned to vehicle $v\in \vV_p$ and $x_{v,p}=0$ otherwise. 
If the platform does not take into account the option of delay-matching of passengers for better pairing opportunities in the future, then the ride-pooling dispatching problem is to solve the following bipartite matching problem at the end of each matching interval:

\begin{equation}
\label{eq:bipartite matching}
  \max_{\textbf{x}}\sum_{p\in \vP}\sum_{v\in \vV_p}  u(v,p) x_{v,p}
\end{equation}
\begin{align}
 \text{s.t.}  &\sum_{v\in \vV_p} x_{v,p}= 1, \forall{p\in\vP}\label{eq: pas_cons}\\
  &\sum_{p\in \vP_v}x_{v,p} \leq 1, \forall{v\in \vV_0\cup\vV_1}\label{eq:veh_cons}\\
  &x_{v,p}\in\{0,1\}, \forall{v\in \vV_p,p\in \vP}\label{eq:integer-cons}.
\end{align}
The objective function (\ref{eq:bipartite matching}) represents the total utility of each matching problem.  
The first constraint, Eq. (\ref{eq: pas_cons}), ensures that every waiting passenger is assigned to a vehicle or kept waiting at the origin. 
The second constraint, Eq. (\ref{eq:veh_cons}), ensures that each vacant or partially occupied vehicle is matched with at most one passenger during each matching round. This bipartite matching problem (\ref{eq:bipartite matching})-(\ref{eq:integer-cons}) is a well-known integer programming problem which can be solved in polynomial time. 

The challenging part of this ride-pooling dispatching problem is to determine the utility of each match, i.e., the value of $u(v,p), \forall v\in \vV_p, p\in \vP$, so as to achieve a long-term goal over a one-day period. 
Various long-term goals have been set in dynamic ride-pooling literature. 
For example, in \cite{xian2020}, the goal is to minimize a weighted sum of total passenger waiting time and trip delay time. 
In \cite{2020Neural}, the goal is to maximize the total number of passengers served. 
Since both the platform's profit and the social welfare of ride-pooling services rely on the distance saving achieved through ride-pooling, this study assumes that the platform's goal is to maximize the total distance saving. 

\subsection{A myopic vehicle dispatching strategy}\label{subsec: A myopic vehicle dispatching strategy}
In this subsection, we present a myopic strategy to determine the utility of each match for maximal total distance saving. Let $l^{pk}(v,o_p)$ represent the pick-up distance between the location of vehicle $v$ and the origin of passenger $p$, and $e(v,p)$ be the resulting distance saving if a vehicle $v\in\vV$ is assigned to pick up the passenger $p\in \vP$. 
If a passenger $p$ is assigned to a partially occupied vehicle $v_{p'}\in \vV_{1,p}$ with passenger $p'$ on board, the distance saving can be immediately calculated by Eq. (\ref{eq:distance saving}): 
\begin{equation}
\begin{aligned}
\label{eq:distance saving}
  e(v_{p'},p) = &l(o_{p'},d_{p'}) + l(o_{p},d_{p}) \\
  &- \min \{l^{fofo}(p',p), l^{folo}(p',p)\}.       
\end{aligned}
\end{equation}
In Eq. (\ref{eq:distance saving}), $o_{p'}$ and $o_p$ respectively indicate the origins of passenger $p'$ and $p$, while $d_{p'}$ and $d_p$ respectively indicate their destinations. $l(o,d)$ denotes the shortest path distance from node $o$ to node $d$ on the road network. 
Thus, $l(o_{p'},d_{p'})$ and $l(o_{p},d_{p})$ correspond to the exclusive riding distances for passengers $p'$ and $p$.
The two terms in the bracket correspond to the total distance of the ride-pooling trip (from the pick-up of the first passenger to the drop-off of both passengers) when passengers are served in first-on-first-off ($l^{fofo}(p',p)$) and first-on-last-off manner ($l^{folo}(p',p)$) respectively.
Figure \ref{figure: Two ways to a delivery a ride-pooling matching pair} illustrates the routes of a vehicle serving a ride-pooling trip in the two different manners, and their trip distance are given by 

\begin{equation}
\begin{aligned}\label{fofo}
  &l^{fofo}(p',p) = l(o_{p'},v_{p'}) + l(v_{p'},o_p)+l(o_p,d_{p'})+l(d_{p'},d_p),\\ 
  &l^{folo}(p',p) = l(o_{p'},v_{p'}) + l(v_{p'},o_p)+l(o_p,d_p)+l(d_p,d_{p'}),
\end{aligned}
\end{equation}
where $v_{p'}$ is a little abused here to represent the location of vehicle $v_{p'}$ when it is dispatched to pick up the second passenger $p$. 

On the other hand, if a passenger is assigned to a vacant vehicle, the instantaneous distance saving from this match is 0, resulting in a net distance saving of $-l^{pk}(v,o_p)$.
So in the myopic strategy, the platform set the utility function of each match as follows:
\begin{align}\label{eq:utility_myopic}
  u(v,p)=\left\{\begin{matrix}
-l^{pk}(v,o_p),  &v\in \vV_{0,p}\\ 
e(v_{p'},p)-l^{pk}(v_{p'},o_p),  &v\in \vV_{1,p}\\
-\infty, &v\notin \vV_{0,p}\cup\vV_{1,p} 
\end{matrix}\right.
\end{align}

\begin{figure}[!t]
    \centering
    \subfloat[First-on-first-off]{\includegraphics[width=3in]{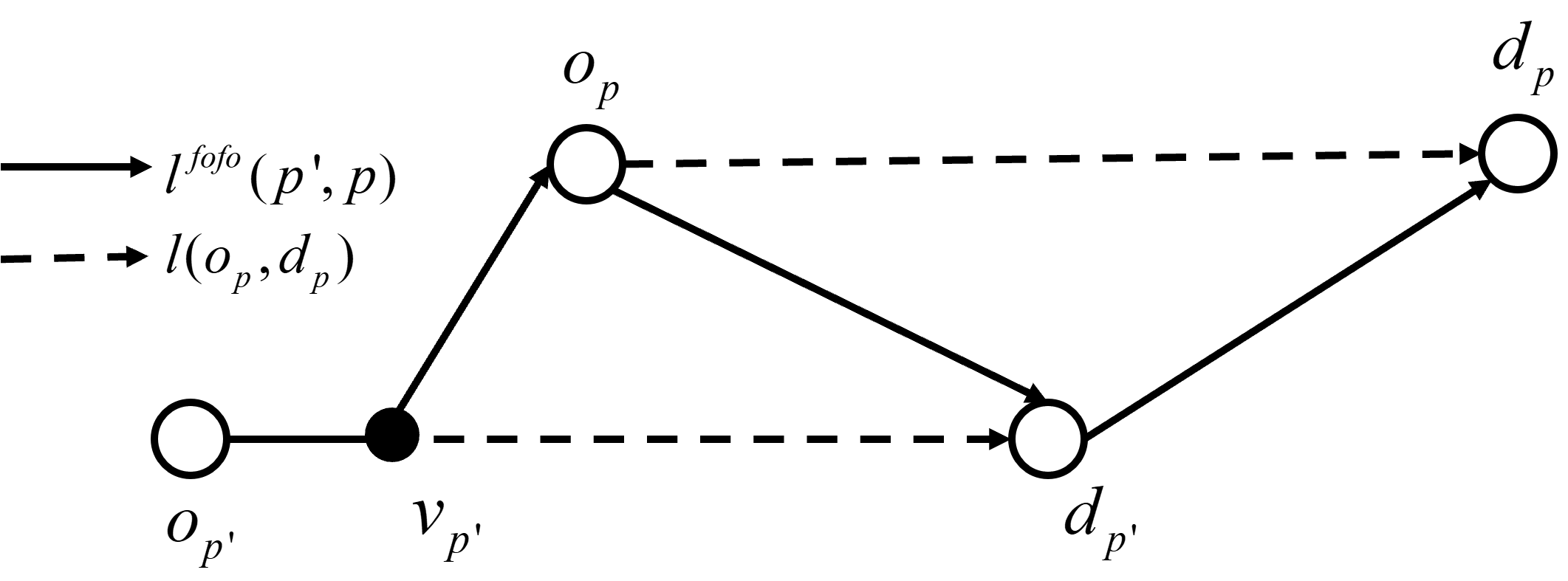}}%
    \label{First-on-first-off}
    
    \subfloat[First-on-last-off]{\includegraphics[width=3in]{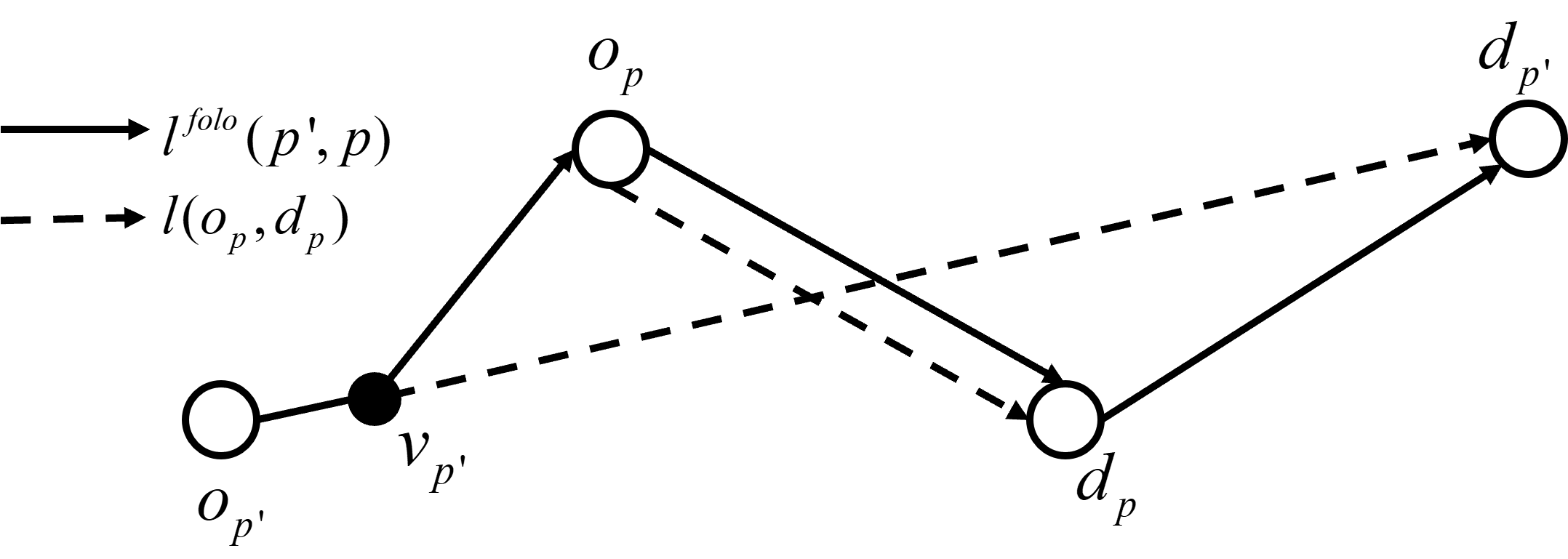}}%
    \label{First-on-last-off}
    \caption{Two ways to serve ride-pooling passengers.}
    \label{figure: Two ways to a delivery a ride-pooling matching pair}
\end{figure}

Apparently, by utilizing Eq. (\ref{eq:utility_myopic}) to determine the utility of each match, the platform will always give matching priorities to partially occupied vehicles that contributes to positive distance saving, and no passenger would be kept waiting at the origin if there are available vacant vehicles nearby. This strategy is subject to shortsighted limitations as it does not take future pairing opportunities into consideration. 

\section{A forward-looking matching strategy for dynamic ride-pooling services}\label{sec:forward-looking matching strategy}
 To make the vehicle dispatching strategy forward-looking, it is essential to have an accurate prediction of the expected distance saving in the future if a passenger is assigned to a vacant vehicle or kept waiting at her/his origin.
 Let $\bar e(p)$ be the expected distance saving of matching passenger $p$ to a vacant vehicle $v\in \vV_{0,p}$, and $\bar e(0,p)$ be the expected distance saving of keeping passenger $p$ waiting at the origin. 
 Note that both $\bar e(p)$ and $\bar e(0,p)$ depend on the characteristics of passenger $p$ only. In this section, we first briefly introduce the prediction method proposed in \cite{wang2021predicting}, then show how we apply the predicted information to determine the value of $\bar e(p)$ and $\bar e(0,p)$ for each passenger $p\in\vP$.
 Finally, the utility function of each match is established, based on the predicted information from a forward-looking perspective. 
 For convenience, Table \ref{Notations} provides a summary of notations used in this paper.
\begin{table*}[]
\caption{List of notations}
\label{Notations}
\centering
\begin{tabular}{ll}
\hline
\textbf{Notations} & \textbf{Indications} \\ 
\hline
$\textbf{P}$ & The set of passengers \\
$p$ & The index of passengers, $p \in \textbf{P}$\\
$o_p$ & The origin of passenger $p$ \\
$d_p$ & The destination of passenger $p$ \\
$l(o,d)$ & The shortest path distance between location $o$ and $d$ \\
$t$ & The matching time index, $t \in [1,T]$ \\
$T$ & The upper limit of matching time index \\
$\Delta t$ & The length of time interval in each batch matching\\
$\textbf{V}_0$ & The set of vacant vehicles \\
$\textbf{V}_1$ & The set of partially occupied vehicles \\
$\{0\}$ & The dummy vehicles associated with the option of keeping passengers waiting \\
$\textbf{V}$ & The set of matching options, $\textbf{V}= \textbf{V}_0 \cup \textbf{V}_1 \cup \{0\}$ \\
$v$ & The index of vehicles, $v \in \textbf{V}$ \\ 
$u(v,p)$ & The utility function of matching vehicle $v$ and passenger $p$\\
$\bar{R}$ & The maximum pick-up distance threshold \\
$\bar{D}$ & The maximum detour distance threshold\\
$K$ & The maximum waiting time threshold for passengers \\
$k$ & The matching rounds of passengers waiting for response, $k \in [1,K/\Delta t]$\\
$\textbf{N}$ & The set of nodes \\
$\textbf{A}$ & The set of links \\
$\textbf{W}$ & The set of OD pairs \\
$w$ & The index of OD pair, $w \in \textbf{W}$ \\
$\textbf{A}_w$ & The set of links traversed by the exclusive-riding path of OD pair $w \in \textbf{W}$ \\
$\lambda_{w}$ & The arrival rate of unpaired passengers of OD pair $w \in \textbf{W}$\\
$s(w)$ & The seeker-state of OD pair $w \in \textbf{W}$, which represents the state of \\ 
 &          newly-appeared passengers in OD pair $w \in \textbf{W}$ \\
$t(a, w), a \in \textbf{A}_w$ & A taker-state, which represent the state of passengers in OD pair 
\\ & $w \in \textbf{W}$ traveling on link $a \in\textbf{A}_w$ \\
$\textbf{T}_{s(w)}$ & The set of matching taker-states of seeker-state $s(w) \in \textbf{S}$ \\
$\textbf{S}_{t(a,w)}$ & The set of matching seeker-states of taker-state $t(a, w) \in \textbf{T}$ \\
$p_{s(w)}$ & The probability that a seeker in state $s(w), w \in \textbf{W}$ gets matched\\
$p_{t(a,w)}$ & The probability that a taker in state $t(a, w) \in \textbf{T}$ gets matched\\
$\rho_{t(a,w)}$ & The probability of having at least one taker in state $t(a, w) \in \textbf{T}$ at any time moment \\
$\eta^{s(w')}_ {t(a,w)}$ & The mean arrival rate of pairing opportunities in state $s(w') \in \textbf{S}_{t(a,w)}$ 
\\ &  for takers in state $t(a, w) \in \textbf{T}$ \\
$\lambda_{t(a,w)}$ & The arrival rate of unpaired passengers into taker-state $t(a, w) \in \textbf{T}$\\
$r_w$ & The presumed response rate of OD pair $w \in \textbf{W}$ \\
$\bar{e}(p)$ & The expected distance saving of matching passenger $p$ with a vacant vehicle $v \in \textbf{V}_{p}$ \\
$\bar{e}(0,p)$ & The expected distance saving of keeping passenger $p$ waiting at the origin \\
$e(v_{p'},p)$ & The deterministic distance saving from matching a partially occupied 
\\ & vehicle $v_{p'}$ with passenger $p$\\
$e(s(w_p))$ & The expected distance saving if passenger $p$ is matched in a seeker-state $s(w_p)$\\
$e(t(a,w_p))$& The expected distance saving if passenger $p$ is matched in a taker-state $t(a,w_p)$\\
$x_{v,p}$ & \begin{tabular}[c]{@{}l@{}}Binary decision variable indicating the matching result between vehicle $v$ and \\ passenger $p$ in the bipartite graph, with 1 for a match and 0 otherwise.\end{tabular} \\
\hline
\end{tabular}
\end{table*}

\subsection{Predicting the expected distance saving of passengers in each state of each OD pair}\label{subsec:Expected distance saving}
In this subsection, we first introduce the prediction method for passenger matching probability in each state of each OD pair developed in \cite{wang2021predicting}.
In a general network $G(\mathbf{N},\mathbf{A})$ with $\mathbf{N}$ being the set of nodes and $\mathbf{A}$ being the set of links, let $\mathbf{W}$ be the set of OD pairs of ride-pooling passengers. 
For each OD pair $w\in \mathbf{W}$, \cite{wang2021predicting} assumes that ride-pooling orders arrive following a Poisson process with a fixed mean demand rate $\lambda_w$, and travel along the same shortest path before getting paired with other orders, where the travel time of each link is assumed to be a constant. 
As a dynamic ride-pooling passenger may be matched at different locations along the route, they define passengers between each OD pair into different states.
\begin{defi}[seeker-state]\label{def:seeker}
    For each OD pair $w\in \mathbf{W}$, passengers waiting at the origin are defined as passengers in seeker-state $s(w)$. 
\end{defi}

\begin{defi}[taker-state]\label{def:taker}
    For OD pair $w\in W$, passengers who are on partially occupied vehicles and travelling on link $a\in \mathbf{A}$ along the route of OD pair $w$ are defined as passengers in taker-state $t(a,w)$.
\end{defi}

For clarity, in the following context of this study, the term "seeker" refers to a passenger in the seeker-state, who is waiting for a response from the ride-pooling platform. And the term "taker" corresponds to a passenger in the taker-state, who has been successfully picked up by a vacant vehicle. 

Let $\mathbf{S}$ be the set of seeker-states, $\mathbf{T}$ be the set of taker-states, and $\mathbf{T}_w$ be the set of taker-states for passengers between OD pair $w\in \mathbf{W}$. 
According to the same pairing conditions as described in subsection \ref{sec:preliminaries} of this paper, the set of taker-states $\mathbf{T}_{s(w)}$ is used to denote the ride-pooling partners that can be matched with the passenger in the seeker-state $s(w)$.
$\mathbf{S}_{t(a,w)}$ denotes the set of seekers which can be paired with the passenger in taker-state $t(a,w)$. 

Following the definition of seeker- and taker-states in Definition \ref{def:seeker} and \ref{def:taker}, a ride-pooling order between OD pair $w\in \mathbf{W}$ would get paired in one of the seeker-state $s(w)$ and taker-state $t(a,w)\in \mathbf{T}_w$, or fail to be paired. 
\cite{wang2021predicting} predicts the pairing probabilities of passengers in each seeker-state and taker-state of each OD pair $w \in \mathbf{W}$ by solving a system of nonlinear equations.
Specifically, they define the following five types of variables and model the interactions among these variables into a system of nonlinear equations: 
\begin{itemize}
\item  $p_{s(w)}$: the probability for a passenger getting paired in seeker-state $s(w)\in \mathbf{S}$; 
\item  $p_{t(a,w)}$: the probability for a passenger getting paired in taker-state $t(a,w)\in \mathbf{T}$; 
\item  ${\rho _{t\left( {a,w} \right)}}$: the probability of having at least one passenger in taker-state $t(a,w)\in \mathbf{T}$; 
\item  ${\eta^{s(w')} _{t\left( {a,w} \right)}}$: the arrival rate of pairing opportunities with seekers $s(w')$ for takers in state $t(a,w)\in \mathbf{T}$; 
\item $\lambda_{t(a,w)}$: the arrival rate of unpaired passengers for taker-state $t(a,w)\in \mathbf{T}$.
\end{itemize}
The detailed formulation of the system of nonlinear equations is omitted in this section (please refer to Appendix \ref{appendix:The system of nonlinear equations problems}).
However, it is important to note that by solving the system of nonlinear equations, the values of all the aforementioned variables for each seeker- and taker-states can be simultaneously determined.
 
Based on the calculated pairing probabilities in each seeker- and taker-states, we then introduce how to determine the expected distance saving $\bar e(p)$ and $\bar e(0,p)$ for any passenger $p$ in the next subsection.
At the end of this subsection, we note that this study assumes the ride-pooling platform adopts a batch matching dispatching strategy, which is different from the first-come-first-serve strategy assumed in \cite{wang2021predicting}. 
As a result, the predicted pairing probability may not be highly accurate for each individual passenger. 
Nevertheless, the numerical experiments in Section \ref{sec: numerical study} demonstrate that the predicted pairing probability can still significantly improve the performance of the dispatching strategy, which underscores the reliability and effectiveness of the proposed forward-looking strategy.

\subsection{Predicting the expected distance saving of matching passengers with vacant vehicles}\label{subsec:expected distance saving of each match}
Following the definition of seeker-state and taker-state in Subsection \ref{subsec:Expected distance saving}, if a passenger $p$ between OD pair $w_p\in W$ is assigned to a vacant vehicle, she/he may get paired in any of the taker-state $t(a,w)\in \mathbf{T}_{w_p}$ (or arrive at the destination without getting matched with anyone). So the expected distance saving that could be brought about by passenger $p$ after assigning it a vacant vehicle, i.e., $\bar e(p)$, is dependent on the passengers' pairing probability as well as the expected distance saving in each taker-state $t(a,w_p)\in \mathbf{T}_{w_p}$. 

Let $\bar e(t(a,w_p))$ be the expected distance saving for takers getting matched in taker-state $t(a,w_p)$. A passenger in taker state $t(a,w_p)$ may get paired with a seeker in any of its matching seeker-state $s(w)\in \mathbf{S}_{t(a,w_p)}$. 
Provided the arrival rate of pairing opportunities from seeker-state $s(w)\in \mathbf{S}_{t(a,w_p)}$ for takers in state $t(a,w_p)$, i.e., $\eta_{t(a,w_p)}^{s(w)}$, we can estimate that for all passengers who are matched in taker-state $t(a,w_p)$, a proportion $\frac{\eta_{t(a,w_p)}}{\sum_{s(w)\in \mathbf{T}_{s(w)}}{\eta_{t(a,w_p)}^{s(w)}}}
$ of them are paired with seekers $s(w)\in \mathbf{S}_{t(a,w_p)}$.  
Let $E(s(w),t(a,w_p))$ be the distance saving, if passenger $p$ in taker-state $t(a,w_p)$ forms a ride-pooling trip with a seeker $s(w)\in \mathbf{S}_{t(a,w_p)}$\footnote{ 
Given the origins and destinations of the seeker and the taker, $E(s(w),t(a,w_p))$ can be easily calculated according to Eq. (\ref{eq:distance saving}).}. Then the expected distance saving for passenger $p$ get paired in taker-state $t(a,w_p)$ can be given by

    \begin{equation}
    \label{eq:distance saving_each taker state}
  \bar e(t(a,w_p))=\sum_{s(w)\in \mathbf{S}_{t(a,w_p)}}\frac{\eta_{t(a,w_p)}^{s(w)}E(s(w),t(a,w_p))}{\sum_{s(w)\in \mathbf{S}_{t(a,w_p)}}\eta_{t(a,w_p)}^{s(w)}},
  \end{equation}

With the expected distance saving in each taker-state given by Eq. (\ref{eq:distance saving_each taker state}), the expected distance saving $\bar e(p)$ of matching passenger $p$ with any vacant vehicle then can be given by

\begin{align}\label{eq:distance saving_taker state}
  \bar e(p)=\sum_{t(a,w_p)\in \mathbf{T}_{w_p}}{\bar e(t(a,w_p))\frac{p_{t(a,w_p)}\lambda_{t(a,w_p)}}{(1-p_{s(w_p)})\lambda_{w_p}}},
\end{align}  
where $(1-p_{s(w_p)})\lambda_{w_p}$ is the total demand rate of passengers who are not paired in seeker-state $s(w)$, $p_{t(a,w_p)}\lambda_{t(a,w_p)}$ is the pairing rate of passengers in taker-state $t(a,w_p)$, and $\frac{p_{t(a,w_p)}\lambda_{t(a,w_p)}}{(1-p_{s(w_p)})\lambda_{w_p}}$ thus gives the passenger pairing probability in taker-state $t(a,w_p)$. 

\subsection{Predicting the expected distance saving of keeping passengers waiting}\label{subsec:expecte distance saving}
If passenger $p$ is kept waiting at the origin in the $k_p$th ($k\in[1,K/ \Delta t]$) round of matching, then she/he may get matched in the seeker-state $s(w_p)$ or be picked up by a vacant vehicle in the future rounds of matching, or quit the dispatching system after $K/ \Delta t$ rounds of unsuccessful matching. 
Let $r_w$ be the presumed response rate for orders between OD pair $w\in \mathbf{W}$, that is, the probability that an order between OD pair $w\in \mathbf{W}$ is assigned to a vacant or partially occupied vehicle in each round of matching, and let $\bar e(s(w_p))$ be the expected distance saving of passenger $p$ getting paired in seeker-state $s(w_p)$.
Provided the expected distance saving of matching $p$ to a vacant vehicle, i.e., $\bar e(p)$, we can approximate the expected distance saving of keeping $p$ waiting at the origin, i.e., $\bar e(0,p)$, by:
\begin{equation}
 \begin{aligned}
       \label{eq:e(0,p)}
  \bar e(0,p)=&\left[1-(1-r_w)^{K/\Delta t-k_p}\right] \\
  & \cdot \left[p_{s(w_p)}\bar e(s(w_p))+(1-p_{s(w_p)})\bar e(p)\right].
\end{aligned}   
\end{equation}
The term inside the first bracket is the probability that passenger $p$ is assigned to a (vacant or partially occupied) vehicle during the remaining $K/\Delta t-k_p$ rounds of matching (before the passenger cancel the order), and the term inside the second bracket is the expected distance saving if passenger $p$ is assigned to a (vacant or partially occupied) vehicle. 

Now we discuss how to calculate $\bar e(s(w_p))$ in Eq. (\ref{eq:e(0,p)}). 
Passenger $p$ in seeker-state $s(w_p)$ may get matched with a passenger in any of its matching taker state $t(a,w)\in \mathbf{T}_{s(w_p)}$. Given the arrival rate of pairing opportunities from seekers in state $s(w_p)$ for takers in state $t(a,w)\in \mathbf{T}_{s(w_p)}$, i.e., $\eta_{t(a,w)}^{s(w_p)}$, and the probability of passenger existences in each taker-state $t(a,w)$, i.e., $\rho_{t(a,w)}$, the expected distance saving for seeker $p$ getting matched while waiting at the origin, or equivalently, $s(w_p)$ can be given by
\begin{equation}\label{eq:distance saving of each seeker state}
    \bar e(s(w_p))=\frac{\sum_{t(a,w)\in \mathbf{T}_{s(w_p)}}E(s(w_p),t(a,w))\rho_{t(a,w)}\eta_{t(a,w)}^{s(w_p)}}{\sum_{t(a,w)\in \mathbf{T}_{s(w_p)}}\rho_{t(a,w)}\eta_{t(a,w)}^{s(w_p)}}.
  \end{equation}
where $\rho_{t(a,w)}\eta_{t(a,w)}^{s(w_p)}$ is the pairing rate between passengers in seeker-state $s(w_p)$ and taker-state $t(a,w)\in \mathbf{T}_{s({w_p})}$.
So the denominator in Eq. (\ref{eq:distance saving of each seeker state}) gives the pairing rate involving passengers in seeker-state $s(w_p)$, and the numerator gives the mean total distance saving involving passengers in the seeker-state $s(w_p)$.

\subsection{Forward-looking utility function design based on the predicted distance saving}\label{subsec:forwad-looking matching strategy}
Given the expected distance saving $\bar e(p)$ and $\bar e(0,p)$in Eqs. (\ref{eq:distance saving_taker state}) and (\ref{eq:e(0,p)}), it is intuitive to set the utility function $u'(v,p)$ into:
\begin{align}\label{eq: combined utility}
  u'(v,p)=\left\{\begin{matrix}
\Bar{e}(p)-l^{pk}(v,o_p),  &v\in \vV_{0,p}\\ 
e(v_{p'},p)-l^{pk}(v_{p'},o_p),  &v\in \vV_{1,p}\\
\Bar e(0,p)-\bar l, &v\in \left\{0\right\}\\ 
-\infty,&v\notin \vV_0\cup\vV_1\cup\left\{0\right\}.
\end{matrix}\right.
\end{align}
where $\bar l$ is the average pick-up distance which can be estimated based on historical data.

To examine whether Eq. (\ref{eq: combined utility}) is a qualified utility function, we propose the following three basic properties that a qualified dynamic vehicle dispatching strategy should meet: 

\textbf{Property 1}. If two passengers $p_1$ and $p_2$ can both match with a vehicle $v\in \vV_{p_1}\cap\vV_{p_2}$ and yield the same distance saving and pick-up distance, then the platform gives matching priority to the passenger who has been waiting longer;

\textbf{Property 2}. If two passengers $p_1$ and $p_2$ who have been waiting for the same rounds of matching and both can be matched with vehicle $v$ under the same pick-up distance, then the platform gives matching priority to the passenger-vehicle matching pair with longer distance saving;

\textbf{Property 3}. When passengers' waiting time at the origin exceeds $K/\Delta t$ rounds of matching, the utility of matching them with any available vehicles exceeds the utility of keeping them waiting at the origin. 

Since the utility $u'(v,p),v\in \vV_{0,p}\cup\vV_{1,p}$ does not change with passenger waiting time, Eq. (\ref{eq: combined utility}) fails to meet Property 1. So we revise Eq. (\ref{eq: combined utility}) into the following form: 
\begin{align}\label{eq: balanced utility}
  u^*(v,p)=\left\{\begin{matrix}
\Bar{e}(p) [\frac{\Bar{e}(p)}{\Bar{e}(p)+l^{pk}(v,o_p)}]\alpha^{k_p}, \quad v\in \vV_{0,p}\\ 
e(v_{p'},p) [\frac{e(v_{p'},p)}{e(v_{p'},p)+l^{pk}(v,o_p)}]\alpha^{k_p}, \quad v\in \vV_{1,p}\\
\Bar{e}(0,p)-\bar l, \quad v\in \left\{0\right\}\\ 
-\infty, \quad v\notin \vV_0\cup\vV_1\cup\left\{0\right\},
\end{matrix}\right.
\end{align}
where $k_p$ is the number of matching rounds that the platform has kept passenger $p$ waiting at the origin, and $\alpha$ is a positive parameter introduced to adjust the priority of each match considering passenger's waiting time. To avoid negative utilities when multiplying with the waiting time coefficient $\alpha^{k_p}$, and meanwhile to ensure that the platform prioritizes passengers with a lower pick-up distance, we reformulate the net distance saving in Eq. (\ref{eq: combined utility}) into a fraction format. 





As shown in the following proposition, the revised utility function $u^*(v,p)$ in Eq. (\ref{eq: balanced utility}) satisfies all of the above-mentioned three basic properties when  $\alpha>1$. 

\begin{prop} Suppose $\alpha>1$, then the above three basic properties of the dynamic ride-pooling matching algorithm can be guaranteed by the revised utility function $u^*(v,p)$. \label{prop1}\end{prop}
\begin{proof}
    See Appendix \ref{appendix:proof of proposition 1}.
\end{proof}

With the revised utility function given by Eq. (\ref{eq: balanced utility}), a forward-looking dispatching strategy can be derived by solving the following bipartite matching problem:
\begin{equation}
\label{eq: forward-looking formulation}
\begin{aligned}
\max_{\mathbf{x}}&\sum_{p\in \vP}\sum_{v\in \vV_p} u^*(v,p)x_{v,p}\\
 \text{s.t.}  &\sum_{v\in \textbf{V}_p} x_{v,p} = 1, \forall{p\in\vP}\\
  &\sum_{p\in \textbf{P}_v}x_{v,p} \leq 1, \forall{v\in \vV_0\cup\vV_1}\\
  &x_{v,p}\in\{0,1\}, \forall{v\in \vV,p\in \vP}
\end{aligned}
\end{equation}

\section{Numerical Studies}\label{sec: numerical study}
In this section, a simulation environment is developed to examine the performance of the proposed forward-looking (FL) matching strategy. To demonstrate its advantages, a comparison is made between the performance of the proposed FL strategy and the following baseline strategies.

- \textbf{Non-pooling matching strategy (NP)}: In this scenario, all passengers opt for non-pooling mode, meaning that the platform matches passengers exclusively with vacant vehicles.

- \textbf{Myopic batch matching strategy (MB)}: Under this baseline, the platform utilizes a myopic batch matching strategy, as given in Section \ref{subsec: A myopic vehicle dispatching strategy}. The utility of each passenger-vehicle matching pair is set according to Eq. (\ref{eq:utility_myopic}).

-\textbf{Request-Trip-Vehicle graph-based strategy (RTV)}: This is a state-of-the-art matching strategy proposed by \cite{doi:10.1073/pnas.1611675114}, which has demonstrated excellent performance in experiments on the road network in New York. To make a fair comparison, we modify the objective of the assignment problem (Step D in \cite{doi:10.1073/pnas.1611675114}) in the RTV strategy into maximizing the total distance saving. 

-\textbf{Forward-looking matching strategy without intended delay-matching (FL-no\_delay)}: In this strategy, passengers would not be intentionally kept waiting at the origin for better pairing opportunities in the future and the utility of each passenger-vehicle matching pair is set according to the following equation:
\begin{align}\label{eq: no delay}
  \hat{u}(v,p)=\left\{\begin{matrix}
\Bar{e}(p)-l^{pk}(v,o_p),   &v\in \vV_0\\ 
e(v_{p'},p)-l^{pk}(v_{p'},o_p),   &v\in \vV_1\\
-\infty, &v\notin \vV_0\cup\vV_1
\end{matrix}\right.
\end{align}
where $\Bar{e}(p)$ is defined in Eq. (\ref{eq:distance saving_taker state}). 

\textbf{The oracle model (Oracle)}: In the oracle model, all ride-pooling trip orders over the whole study periods are assumed to be known in advance. So it serves as a theoretical upper bound for ride-pooling efficiency. By comparing the results obtained from the oracle model with those of the FL strategy, the competitive ratio can be evaluated, describing how close the performance of the proposed FL strategy is to the optimal solution. In this study, the optimal matching results in the oracle model is obtained by solving the following linear programming problem:

\begin{equation}
\label{eq: oracle model}
  \max_{\textbf{x}} \sum_{i\in \mathbf{P}^T}\sum_{j\in \mathbf{P}^T}  e(i,j) x_{i,j}
\end{equation}
\begin{align}
 \text{s.t.}  &\sum_{j\in \mathbf{P}^T} x_{i,j} = 1, \forall{i\in \mathbf{P}^T} \label{eq: i_cons}\\
  &\sum_{i\in \mathbf{P}^T} x_{i,j} = 1, \forall{j\in \mathbf{P}^T} \label{eq: j_cons}\\
  & x_{i,j} = x_{j,i}, \forall{i,j\in \mathbf{P}^T} \label{eq:flow_cons}\\
  & x_{i,j}\in\{0,1\}, \forall{i,j\in \mathbf{P}^T}.\label{eq:integer-cons-1}
\end{align}
where $\mathbf{P}^T$ is the set of all ride-pooling orders over the simulation period $[1,T]$. The objective function (\ref{eq: oracle model}) aims to maximize the total distance saving. The binary decision variable $x_{i,j}$ is 1 if passenger $i\in \mathbf{P}^T$ is paired with passenger $j\in \mathbf{P}^T$, and 0 otherwise. The resulting distance saving $e(i,j)$ is calculated following Eq. (\ref{eq:distance saving}), and it is important to note that $e(i,i)=0$. The constraints Eq. (\ref{eq: i_cons}) and Eq. (\ref{eq: j_cons}) ensure that every passenger is matched to at most one another passenger, and constraint Eq. (\ref{eq:flow_cons}) ensures that if passengers $i$ is matched with passenger $j$, then passenger $j$ is also matched with passenger $i$. In this study, the integer programming problem is solved using Gurobi.

In the following sections, the setup of simulation experiments will be introduced, followed by the presentation of the simulation results.

\subsection{Setup of the simulation experiments}\label{subsec: Setup of the simulation experiments}
The simulation experiments use an open dataset of ride-hailing orders in Haikou, a city in China, provided by Didi Chuxing. The dataset includes ride-hailing orders recorded from May 1 to 21, 2017.
The experiments simulate the occurrence, matching, and movements of these orders on the road network of Haikou, China.

The simulation adopts a default matching time interval of 10 seconds ($\Delta t=10$s) and we set the number of drivers according to a default order-driver ratio of 100:25. 
The initial locations of drivers on the road network are randomly generated.
All drivers are assumed to choose the shortest path traveling between any two nodes on the network, with a constant speed 30km/hr. 
After dropping off all on-board passengers, a fully or partially occupied vehicle becomes vacant and randomly travels through the road network until it is assigned to a passenger. 
Passengers' maximum waiting time for response follows a normal distribution with a mean of 90s and a standard deviation of 10s. 
The value of parameter $\alpha$ in Eq. (\ref{eq: balanced utility}) and the presumed response rate $r_w$ in Eq. (\ref{eq:e(0,p)}) are set to be 1.01 and 0.75, respectively. 
The pairing conditions are set to be $\Bar{R}=3000$m, $\Bar{D}=3000$m respectively.
In this study, we assume the objective of the dispatching problem is to maximize total distance saving, so the platform's pricing strategy is not considered to impact the matching strategy. 
Nevertheless, to give a rough idea of the impacts of the proposed FL matching strategy on the platform's profit, we set the platform's pricing strategy as follows: passengers are charged 5 Yuan/km for solo-ride distance, and 3.5 Yuan/km for shared distance; drivers are paid 2 Yuan/km for every occupied driving distance.

The FL strategy is implemented in the following procedure. First, to conduct the offline prediction of expected distance saving according to the method introduced Section \ref{sec:forward-looking matching strategy}, we first statistically calculate the average demand rate of each OD pair for each unit study period (that is, one hour in our experiment) based on the real-world order dataset recorded between May 1 to 21, 2017.
These average demand rates for all OD pairs are then input into the system of nonlinear equations (\ref{eq:p_sw})-(\ref{eq:lambda_t(a,w)}) to obtain the passenger pairing probabilities in each seeker- and taker-state. And we store the solutions of Eqs.  (\ref{eq:p_sw})-(\ref{eq:lambda_t(a,w)}) for online matching. 
In the real-time matching problem, by retrieving the stored predicted information, we calculate the expected distance saving $\bar e(p)$ and $\bar e(0,p)$ according to Eqs. (\ref{eq:distance saving_taker state}) and (\ref{eq:e(0,p)}) for each passenger $p$. And finally, we solve the bipartite matching problem (\ref{eq: forward-looking formulation}) with utility of each match set according to Eq. (\ref{eq: balanced utility}) to generate the FL matching result. 

\subsection{Simulation results}\label{subsec: Simulation results}

Under the above setup, Table \ref{tab2: Performance of different models in system metrics} provides a summary of the simulation results for the proposed strategy and baseline strategies using the ride-hailing orders between 8:00 am and 9:00 am on May 1st, 2017 in the open dataset of Haikou.
As observed from the table, in comparison with the non-pooling (NP) strategy, pooling passengers together by any of the four matching strategies (i.e., MB, RTV, FL-no\_delay and FL) can lead to a higher order response rate. 

Among the four ride-pooling matching strategies (MB, RTV, FL-no\_delay, and FL), the RTV strategy proposed by \cite{doi:10.1073/pnas.1611675114} outperforms the MB strategy in terms of response rate, total distance saving and passenger detour distance, but under-performs the FL-no\_delay and FL strategy in terms of both average detour distance and total distance saving. 
In comparison with the RTV strategy, the proposed FL strategy (FL-no\_delay strategy) can generate 31.7\% (14.2\%) more total distance saving and reduces the average passenger detour distance by 18\% (5.4\%). 
Furthermore, such longer total distance saving is not achieved by serving more orders. 
As shown in Table \ref{tab2: Performance of different models in system metrics}, the response rates under both the FL and FL-no\_delay strategies are smaller than those under the MB and RTV strategies. 
This highlights the effectiveness of the proposed utility function in forfeiting low-quality matches and prompting high-quality matches.  

It is worthwhile to emphasize that the lower response rate to dynamic ridepooling orders may not be a disadvantage for existing ride-hailing platforms. 
Although fewer orders are responded, the FL strategy yields the highest total distance saving and lowest detour distance for passengers, by deliberating choosing orders with high matching potentials to respond. 
Most existing ride-hailing platforms, e.g., Didi and Uber, operate both ride-pooling and non-pooling services, which compete with each other for vehicle supply. 
The matching results generated by the FL strategy for ride-pooling service could help platforms to better allocate their supply resources and ensures that ride-pooling is encouraged only for orders that have high pairing potentials. 

To further assess the effectiveness of the intended delay-matching in the FL strategy, a comparison is made between the FL strategy and the FL-no\_delay strategy. 
The results presented in Table \ref{tab2: Performance of different models in system metrics} show that implementing the FL-no\_delay strategy results in a higher response rate (0.9 v.s. 0.81) and a lower average response time (6s v.s. 30s).
But the pairing ratio is reduced to 0.18, the total distance saving decreases to 137km, and the average detour distance is increased to 855m. 
These findings indicate that the intended delay-matching in the FL strategy improves matching quality, but at the cost of longer waiting time and lower response rate.

In Table \ref{tab2: Performance of different models in system metrics}, Oracle-1 denotes the oracle model described in Section \ref{sec: numerical study}, while Oracle-2 indicates another oracle model that aims to maximize the number of successfully paired passengers. 
As shown in Table \ref{tab2: Performance of different models in system metrics}, the maximal total distance saving that can be achieved in Oracle-1 is 250km, and the proposed FL strategy can achieve 63.2\% of the optimal total distance saving under Oracle-1, with the resulting average detour distance being 211m less than that under the Oracle-1 model. 
With an objective of maximizing the total number of successfully paired passengers, Oracle-2 model achieves the maximal paring ratio 82\%, but the resulting average detour distance is 297m longer than that under Oracle-1 model, and the total distance saving (50km) is the lowest among all the presented matching strategies. 
This highlights the fact that not all passengers are suitable candidates for ride-pooling services. 
Instead of responding and paring as many orders as possible, filtering orders with poor matching potentials through the proposed FL strategy could lead to simultaneous and significant improvements of both system performance and user experience. 

\begin{table*}[]
\caption{Performance of different vehicle dispatching strategies}
\label{tab2: Performance of different models in system metrics}
\centering
\resizebox{\textwidth}{!}{%
\begin{tabular}{cccccccl}
\hline
Metrics                                                                                                        & NP   & MB   & RTV    & FL   & FL-no\_delay & Oracle-1 & Oracle-2 \\ \hline
Response rate                                                                                                  & 0.74 & 0.9  & 0.94   & 0.81 & 0.9          & 0.81  & 0.84    \\
Pairing ratio                                                                                                  & /    & 0.2  & 0.2    & 0.27 & 0.18          & 0.66  & 0.82  \\
Platform profit                                                                                                & 7232 & 4572 & 4633   & 4438 & 4529         & 5205  & 4830  \\
\begin{tabular}[c]{@{}c@{}}Average response time,\\  $avg\_resp\_time$ (s)\end{tabular}                        & 16   & 5    & 5      & 30   & 6            & 5     & 5     \\
\begin{tabular}[c]{@{}c@{}}Average pick-up time, \\ $avg\_pk\_time$ (s)\end{tabular}                           & 168  & 107  & 122    & 111  & 103          & 174   & 197   \\
\begin{tabular}[c]{@{}c@{}}Average detour distance,\\  $avg\_detour$(m)\end{tabular}                          & /    & 912  & 892   & 731  & 855          & 942   & 1239  \\
\begin{tabular}[c]{@{}c@{}}Average shared distance, \\ $avg\_share$(m)\end{tabular}                           & /    & 4820 & 4773   & 4188 & 4913         & 5562  & 2716  \\
\begin{tabular}[c]{@{}c@{}}Total driving distance with passenger(s) \\on board, $dist\_total$ (km)\end{tabular} & 3333 & 3796 & 3984   & 3310 & 3772         & 3499  & 3561  \\
\begin{tabular}[c]{@{}c@{}}Total distance saving  due to\\ ride-pooling, $dist\_save$ (km)\end{tabular}          & /    & 111  & 120    & 158  & 137          & 250   & 51  \\ \hline
\end{tabular}%
}
\end{table*}

\begin{figure}[!t]
    \centering
    \includegraphics[width=3in]{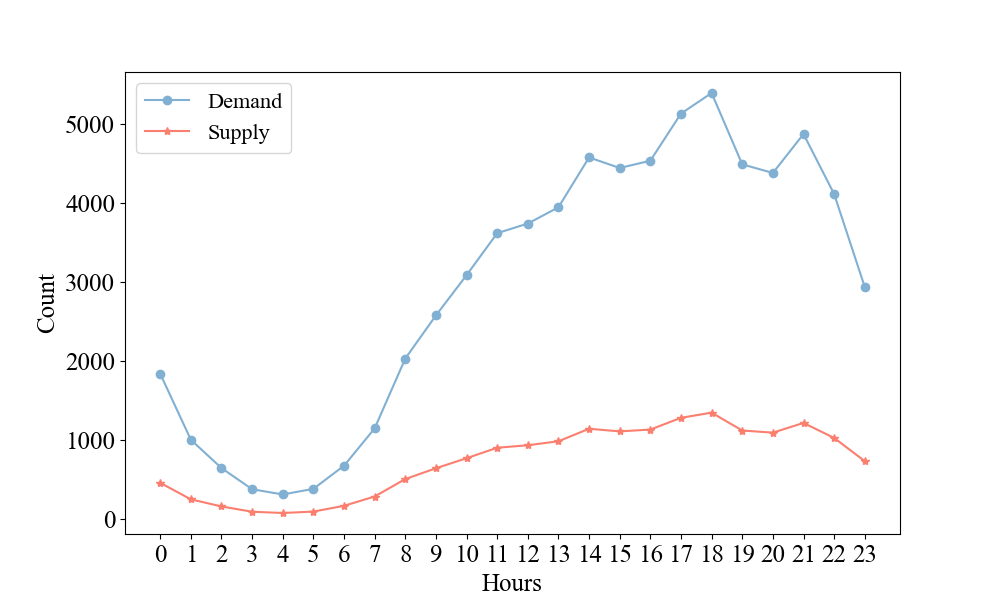}%
    
    \caption{Variation of demand and vehicle supply in one day}
    \label{figure: The quantity of demand-supply}
\end{figure}

For the MB, RTV, and FL strategies, this study further simulates their performances against the real-world orders from 0:00 to 24:00 on 1st May, 2017. 
The variation of demand and vehicle supply over the day are presented in Figure \ref{figure: The quantity of demand-supply}, while the corresponding simulation results are depicted in Figure \ref{figure: Comparison of different method in different simulation hours}.
As we can see from Figure \ref{figure: The quantity of demand-supply}, the ride-hailing demand increases quickly from 5:00 to 10:00, and then keeps at a high level from 10:00 to 24:00.
Under all the three strategies, the variation trend of the total distance saving is consistent with the variation of ride-hailing demand, as shown in Figure \ref{figure: Comparison of different method in different simulation hours} (c). The average shared distance shown in Figure \ref{figure: Comparison of different method in different simulation hours} (a) peaks at 6 AM when the trip demand is relatively low\footnote{This is because a large proportion of trips during this period are destined for the airport, and most airport trips exceed 10km.}.
Furthermore, as we can see from Figure \ref{figure: Comparison of different method in different simulation hours}, the FL strategy always outperforms the RTV and MB strategies in terms of total distance saving. The magnitude of this increase is particularly notable when there is a high demand and a large demand-supply ratio from 10:00 to 24:00. During that period, the total distance saving generated by the FL strategy is almost twice that of the RTV strategy.
Meanwhile, the detour distance of the FL strategy is also the lowest during the simulation period, as shown in Figure \ref{figure: Comparison of different method in different simulation hours} (b).
Despite the longer total distance saving, the response rate is lowest in the FL strategy, as shown in Figure \ref{figure: Comparison of different method in different simulation hours} (d). 
While the lower response rate could be considered a disadvantage of the proposed strategy when applied to a platform with solely ride-pooling service, it becomes an advantage when adopted by platforms such as Didi and Uber, which operate both ride-pooling and non-pooling services simultaneously.
The significantly longer total distance saving achieved through serving much less passengers highlights the accuracy of the proposed strategy in identifying orders with high pairing potentials.

\begin{figure*}[!t]
\centering
\subfloat[Average shared distance for passengers]{\includegraphics[width=2.5in]{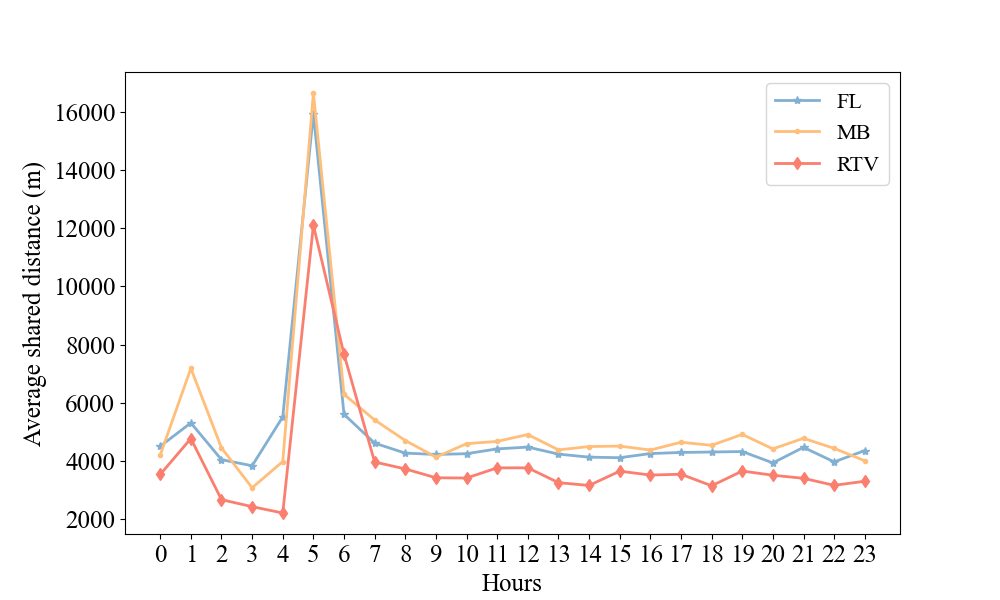}%
}
\hfil
\subfloat[Average detour distance for passengers]{\includegraphics[width=2.5in]{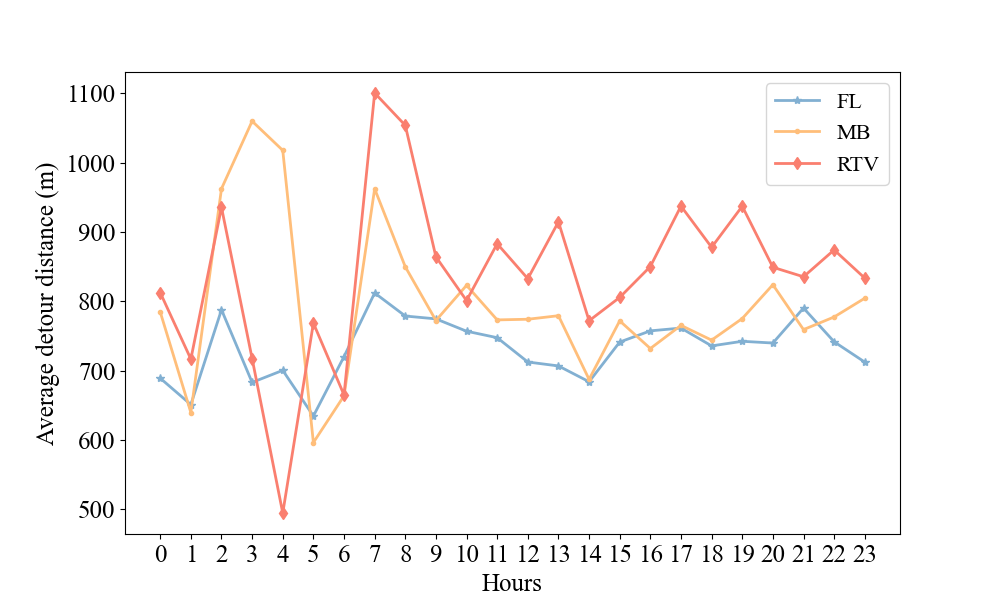}%
}
\quad
\subfloat[Total distance saving due to ride-pooling]{\includegraphics[width=2.5in]{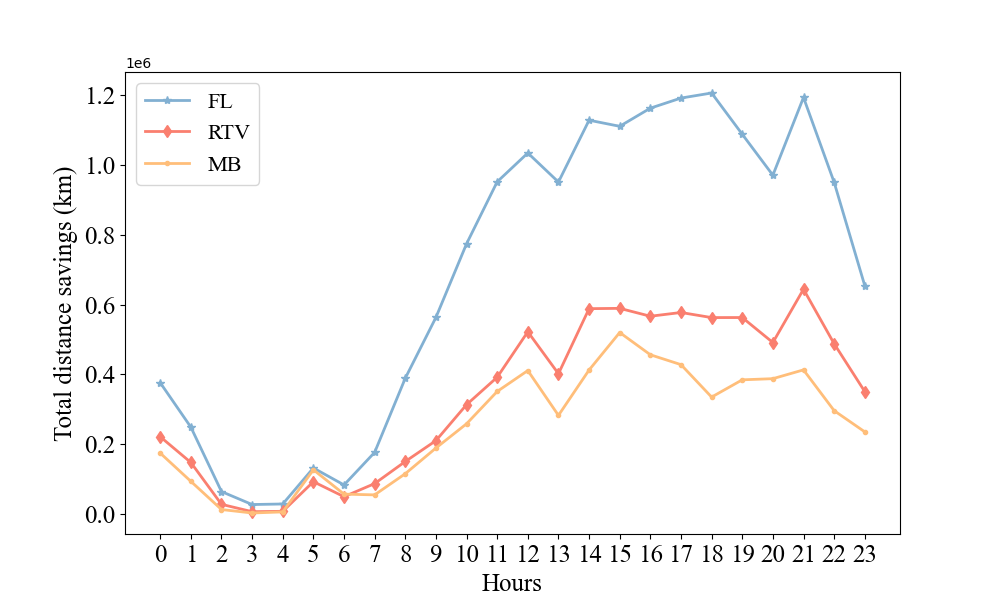}%
}
\hfil
\subfloat[Response rate]{\includegraphics[width=2.5in]{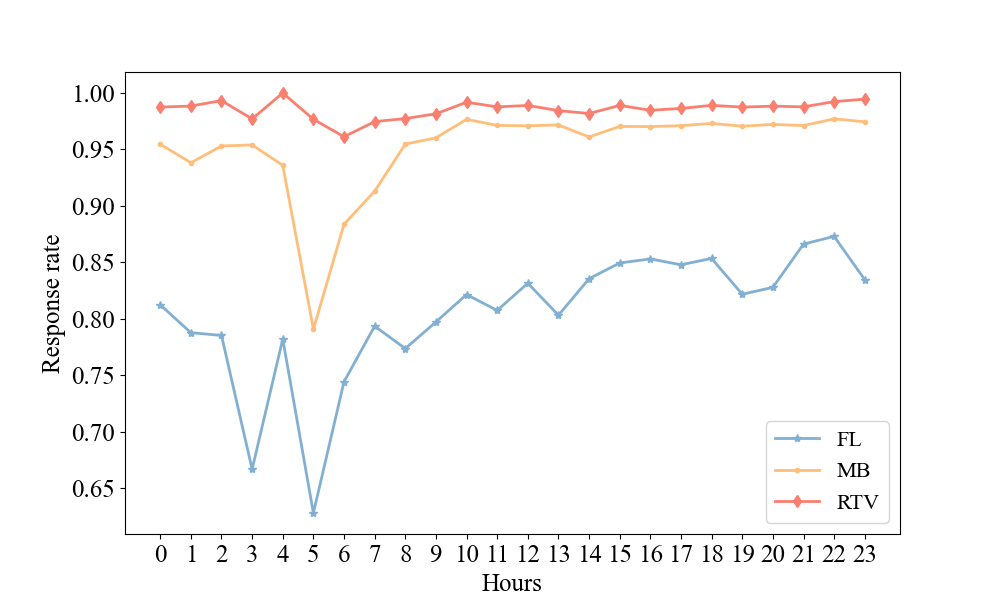}%
}
    \caption{Comparison of the performance under the MB, RTV and FL strategies in different hours in the Haikou experiment.}
    \label{figure: Comparison of different method in different simulation hours}
\end{figure*}

\subsection{Sensitivity analysis}\label{subsec:What-if analysis}

In this subsection, we further conduct sensitivity analysis to assess the performance of the FL strategy under diverse scenarios. 

\subsubsection{Influence of the mean of passengers' maximum waiting time}\label{subsubsec: Influence of the average response time}
The ride-pooling service quality is sensitive to passengers' maximum waiting time, particularly for low demand density \citep{stiglic2016making}. Commonly used maximum waiting time in the literature
include 300s, 600s, 900s, and so on \citep{bilali2020analytical,engelhardt2019quantifying,zwick2021ride}. Table \ref{tab:Sensitive analysis of parameter K} presents the sensitivity analysis results of dispatching performance under different matching strategies as the parameter $K$ varies from 60s to 600s. 

As we can see from the table, as passengers' maximum waiting time $K$ increases, both the response rate and the average response time increase, especially under our proposed FL strategy. When the maximum passenger waiting time exceeds 240s (4min), the difference between the response rates under our FL strategy and the other two strategies become less than 5\%. On the other hand, in terms of total distance saving, the proposed FL strategy keeps outperforming both the MB and the RTV strategies under all values of $K$. 
Both the MB strategy and the RTV strategy are not sensitive to passengers' maximum waiting time since they always prompt passenger-vehicle matching. 
However, as $K$ increases from 60s to 600s, the response rate of the proposed FL strategy increases from 0.77 to 0.92, the pairing ratio increases from 0.24 to 0.31, and the total distance saving increases from 150km to 207km, while the average detour distance decreases from 776m to 728m.
In conclusion, when the average maximal waiting time for passengers is longer, the superiority of our FL strategy in generating longer total distance saving and a higher response rate would be more obvious. 

\begin{table*}[]
\centering
\caption{Sensitive analysis of parameter $K$}
\label{tab:Sensitive analysis of parameter K}
\resizebox{\textwidth}{!}{%
\begin{tabular}{c|ccc|ccc|ccc|ccc|ccc}
\hline
\multirow{2}{*}{$K$} & \multicolumn{3}{c|}{$avg\_detour\_dist$(m)} & \multicolumn{3}{c|}{$dist\_save$(km)} & \multicolumn{3}{c|}{$avg\_resp\_time$(s)} & \multicolumn{3}{c|}{Response rate} & \multicolumn{3}{c}{Pairing ratio} \\
                     & MB            & RTV          & FL           & MB          & RTV        & FL         & MB          & RTV          & FL           & MB         & RTV       & FL        & MB        & RTV       & FL        \\ \hline
60                   & 879           & 922          & 776          & 108         & 118        & 150        & 5           & 5            & 25           & 0.88       & 0.93      & 0.77      & 0.19      & 0.13      & 0.24      \\
90                   & 912           & 894          & 731          & 110         & 120        & 158        & 6           & 5            & 30           & 0.9        & 0.94      & 0.81      & 0.2       & 0.14      & 0.27      \\
120                  & 889           & 821          & 793          & 115         & 150        & 167        & 6           & 6            & 33           & 0.91       & 0.94      & 0.86      & 0.2       & 0.14      & 0.28      \\
180                  & 896           & 821          & 765          & 117         & 150        & 193        & 7           & 6            & 37           & 0.91       & 0.94      & 0.89      & 0.18      & 0.14      & 0.29      \\
240                  & 867           & 821          & 756          & 125         & 150        & 198        & 9           & 6            & 42          & 0.92       & 0.94      & 0.90      & 0.18      & 0.14      & 0.28      \\
300                  & 906           & 834          & 751          & 134         & 155        & 200        & 11          & 7            & 44          & 0.93       & 0.95      & 0.90      & 0.2       & 0.15      & 0.30      \\
400                  & 827           & 834          & 749          & 139         & 155        & 207        & 11          & 7            & 47          & 0.93       & 0.95      & 0.92      & 0.2       & 0.15      & 0.31      \\
500                  & 874           & 834          & 763          & 143         & 155        & 207        & 14          & 7            & 51          & 0.93       & 0.95      & 0.92      & 0.19      & 0.15      & 0.31      \\
600                  & 883           & 834          & 728          & 147         & 155        & 207        & 15          & 7            & 54          & 0.93       & 0.95      & 0.92      & 0.18      & 0.15      & 0.31      \\ \hline
\end{tabular}%
}
\end{table*} 

\subsubsection{Influence of the presumed response rate}\label{subsubsec: Influence of the average response rate}
In Eq. (\ref{eq:e(0,p)}), the utility of keeping passenger waiting is influenced by a hyper parameter: the presumed response rate $r_w$. 
As response rate varies with demand-supply conditions as well as the platform's matching strategy, setting an accurate parameter $r_w$ in Eq. (\ref{eq:e(0,p)}) is difficult.
To examine the impact of the presumed response rate $r_w$ on the performance of the FL strategy, Table \ref{tab: Sensitive analysis of parameter r_w} presents the results when $r_w$ increases from 0.1 to 0.9, with other parameters remaining the same as described in subsection \ref{subsec: Setup of the simulation experiments}.
The simulation results demonstrate that the proposed FL strategy is not very sensitive to the presumed response rate $r_w$. When $r_w$ changes from 0.1 to 0.9, the actual response rate changes only slightly from 0.87 to 0.82, and the pairing ratio changes slightly from 0.25 to 0.26. 
The change in total distance saving is less than 2.5\% (from 160km to 156km), and the change in average detour distance is less than 3.4\%.

\begin{table*}[]
\caption{Sensitive analysis of parameter $r_w$}
\label{tab: Sensitive analysis of parameter r_w}
\centering
\resizebox{\textwidth}{!}{%
\begin{tabular}{ccccccc}
\hline
$r_w$ & $avg\_share\_dist$(m) & $avg\_detour\_dist$(m) & $dist\_save$(km) & $avg\_resp\_time$(s) & Response rate & Pairing ratio \\ \hline
0.1   & 4539                  & 773                    & 160              & 21                   & 0.87          & 0.25          \\
0.2   & 4623                  & 772                    & 162              & 27                   & 0.86          & 0.27          \\
0.3   & 4660                  & 733                    & 168              & 29                   & 0.85           & 0.26          \\
0.4   & 4034                  & 764                    & 158              & 30                   & 0.84          & 0.26          \\
0.5   & 4444                  & 778                    & 167              & 31                   & 0.84          & 0.29          \\
0.6   & 4206                  & 786                    & 165              & 31                   & 0.83          & 0.29          \\
0.7   & 4307                  & 776                    & 162              & 30                   & 0.83          & 0.27          \\
0.8   & 4188                  & 731                    & 158              & 30                   & 0.81          & 0.27          \\
0.9   & 4280                  & 746                    & 156              & 30                   & 0.82          & 0.26          \\ \hline
\end{tabular}%
}
\end{table*}

\subsubsection{Influence of the matching time interval}\label{subsubsec: Influence of the matching time interval}
The matching time interval $\Delta t$ is another hyper-parameter in the dynamic ride-pooling dispatching problem. 
As pointed out in \cite{yang2020optimizing}, the matching time interval has an obvious impact on the dispatching performance. 
So we conduct a sensitivity analysis when $\Delta =2s,10s,30s$ and $60s$. 
The simulation results, presented in Table \ref{tab: Sensitive analysis of parameter Delta t}, show that as the matching time interval $\Delta t$ increases, all the three matching strategies show improvements in total distance saving and the ride-pooling pairing ratio, together with increases in average response time and average passenger detour. 
Regardless of the value for $\Delta t$, the proposed FL strategy consistently outperforms both the MB and RTV strategies in total distance saving and average passenger detour distance, but the degree of superiority decreases with longer $\Delta t$. 
This is not surprising, as the benefits of being forward-looking becomes less significant with more matching opportunities being encountered within a longer matching time interval $\Delta t$.

\begin{table*}[]
\centering
\caption{Sensitive analysis of parameter $\Delta t$}
\label{tab: Sensitive analysis of parameter Delta t}
\resizebox{\textwidth}{!}{%
\begin{tabular}{c|ccc|ccc|ccc|ccc|ccc}
\hline
\multirow{2}{*}{$\Delta t$} & \multicolumn{3}{c|}{$avg\_detour\_dist$(m)} & \multicolumn{3}{c|}{$dist\_save$(km)} & \multicolumn{3}{c|}{$avg\_resp\_time$(s)} & \multicolumn{3}{c|}{Response rate} & \multicolumn{3}{c}{Pairing ratio} \\
                           & MB            & RTV          & FL           & MB          & RTV        & FL         & MB           & RTV          & FL          & MB         & RTV       & FL        & MB        & RTV       & FL        \\ \hline
2                          & 820           & 756          & 714          & 108         & 114        & 141        & 2            & 1            & 21          & 0.91       & 0.95      & 0.85      & 0.14      & 0.11      & 0.20      \\
10                         & 912           & 894          & 731          & 111         & 120        & 158        & 5            & 5            & 30          & 0.9        & 0.94      & 0.81      & 0.2       & 0.2       & 0.27      \\
30                         & 966           & 842          & 795          & 223         & 264        & 309        & 16           & 16           & 47          & 0.88       & 0.94      & 0.78      & 0.34      & 0.25      & 0.45      \\
60                         & 922           & 838          & 803          & 323         & 335        & 338        & 31           & 31           & 70          & 0.84       & 0.92      & 0.72      & 0.47      & 0.41      & 0.57      \\ \hline
\end{tabular}%
}
\end{table*}

\subsubsection{Influence of the order-driver ratio}\label{subsubsec: Influence of the order-driver ratio}
In Section \ref{subsec: Setup of the simulation experiments}, the order driver ratio for the simulation is set as 100:25.
In this subsection, we further assess the performance of the proposed strategy in comparison the MB, RTV, and FL matching strategies under a higher order-driver ratio (100:10) scenario.
The results are presented in Table \ref{tab6: Influence of different order-driver ratio}.
The FL strategy consistently outperforms other strategies in terms of total distance saving and average passenger detour, regardless of the driver supply.

\begin{table*}[]
\caption{Influence of different order-driver ratio}
\label{tab6: Influence of different order-driver ratio}
\resizebox{\textwidth}{!}{%
\begin{tabular}{llllllll}
\hline
Order-driver ratio      & Strategy & $avg\_detour\_dist$(m) & $dist\_save$ (km) & $avg\_resp\_time$ (s) & Platform profit & Response rate & Pairing ratio \\ \hline
\multirow{3}{*}{100:10} & MB     & 913                     & 109              & 13                    & 3781            & 0.71          & 0.36          \\
                        & RTV & 923                     & 118              & 11                    & 4128            & 0.86          & 0.35          \\
                        & FL     & 756                     & 148              & 37                    & 3616            & 0.69          & 0.32          \\ \hline
\multirow{3}{*}{100:25} & MB     & 912                     & 111              & 5                     & 4572            & 0.9           & 0.2           \\
                        & RTV & 894                     & 120              & 5                     & 4633            & 0.94          & 0.2           \\
                        & FL     & 731                     & 158              & 30                    & 4238            & 0.81          & 0.27          \\ \hline
\end{tabular}%
}
\end{table*}

\section{Conclusion}\label{sec: conclusion}
In this study, we propose a forward-looking dispatching strategy to maximize total distance saving through dynamic ride-pooling service. 
Considering the pairing opportunities brought about by future orders, our method first predicts the expected distances saving for passengers if they were kept waiting or assigned to vacant vehicles based on historical demand distributions over the network, and then set the utility of each match according to the prediction to make forward-looking decisions. 
A simulation environment is developed based on the Haikou road network, and the performance of the proposed forward-looking strategy is compared to four benchmark strategies using an open dataset of ride-hailing orders in Haikou, China. 
As shown by simulation results, our proposed strategy generates the highest total distance saving and lowest average passenger detour distance.
Specifically, compared to the state-of-the-art matching strategy proposed by \cite{doi:10.1073/pnas.1611675114}, the proposed strategy leads to 31\% more total distance saving and 18\% less average passenger detour distance.

This study contributes to the development of efficient vehicle dispatching strategies in ride-pooling services, and highlights the benefits of taking future matching opportunities into consideration.
With a reliable prediction of the expected distance saving of different passengers under different matching options, the platform can effectively improve the system performance (in terms of longer total distance saving) as well as user experience (in terms of shorter passenger detour distance) by deliberately giving up some current pairing opportunities, and keeping some passengers waiting a little bit longer. 
Future work could explore the performance of the proposed strategy under varying road network conditions, and investigate ways to further improve the accuracy of the prediction information.

\appendices

\section{The system of nonlinear equations problems in \cite{wang2021predicting} for predicting the pairing probability in each seeker- and taker-states}\label{appendix:The system of nonlinear equations problems}

To ease understanding, we provide a brief introduction of the model that characterize the complex interactions among the five types of variables defined in subsection \ref{subsec:Expected distance saving}.
For a more in-depth discussion, please refer to \cite{wang2021predicting}. 

The model assumes that the platform adopts the first-come-first-serve rule when searching for matching takers of each seeker and prompts a immediate match if it finds a matching taker for the seeker.
Therefore, a seeker $p$ would be immediately paired if there are passengers in one of its matching taker-states $t\left( {a,w} \right) \in {\mathbf{T}_{s(w_p)}}$. 
The pairing probability $p_{s(w_p)}$ in seeker-state $s(w_p)\in \mathbf{S}$ is thus dependent on the probability ${\rho _{t\left( {a,w} \right)}}$ of having at least one takers in state $t\left( {a,w} \right) \in {\mathbf{T}_{s(w_p)}}$ at any moment:
\begin{equation}
\label{eq:p_sw}
p_{s(w_p)} =
\begin{cases}
1- {\textstyle \prod\limits_{t(a,w) \in \mathbf{T}_{s(w_p)}}\left ( 1- \rho_{t(a,w )} \right )},\text{if} \ \mathbf{T}_{s(w_p)}\ne \emptyset
\\0, \text{otherwise}
\end{cases}.
\end{equation}

On the other hand, a passenger $p$ in taker-state $t(a,w_p)$ may be paired if there are pairing opportunities arriving during its stay in state $t(a,w_p)$. 
So $p_{t(a,w_p)}$ is dependent on the aggregate arrival rate of pairing opportunities for passengers in taker-state $t(a,w_p)\in \mathbf{T}$:

%

\begin{equation}\label{eq:p_t(a,w)}
p_{t(a,w_p)} =
\begin{cases}
= 1-\text{exp} \left( - \eta_{t(a,w_p)}\overline{\tau }_{t(a,w_p)}\right),\text{if} \ \eta_{t(a,w_p)} > 0
\\0, \text{if} \ \eta_{t(a,w_p)} = 0
\end{cases},
\end{equation}
where $\eta_{t(a,w_p)}$ is the sum of the arrival rates of pairing opportunities ${\eta _{t\left( {a,w_p} \right)}^{s(w)}}$ from all of its matching seeker-states $s\left( w \right) \in {S_{t(a,w_p)}}$:
\begin{equation}\label{eq:eta_t(a,w)}
\eta _{t(a,w_p)} =
\begin{cases}
\sum\limits_{s(w) \in \mathbf{S}_{t(a,w_p)}} \eta _{t(a,w_p)}^{s(w)},\text{if} \ \mathbf{S}_{t(a,w_p)} \ne \emptyset
\\0, \text{otherwise}
\end{cases}.
\end{equation}

A newly appeared passenger in seeker-state $s(w)$ can be counted as a pairing opportunity for passengers in taker-state $t(a,w_p)$ if there is no passenger in any of the taker-state $t(a',w')\in \mathbf{T}_{s(w)^{\succ t(a,w_p)}}$.
Here $\mathbf{T}_{s(w)^{\succ t(a,w_p)}}$ defines the set of all taker-states that have higher matching priorities with the seeker $s(w)$ in comparison with $t(a,w_p)$, and is predetermined according to the matching quality (i.e., distance saving, pick-up distance, passenger detour distance) of each seeker-taker pair. 
So ${\eta _{t\left( {a,w_p} \right)}^{s(w)}}$ can be approximated by Eq. (\ref{eq:eta_t(a,w)^s(w)}).

\begin{figure*}[!h]
\begin{equation}\label{eq:eta_t(a,w)^s(w)}
\eta _{t(a,w_p)}^{s(w)} =
\begin{cases}
\lambda_{w}, \text{if} \ s(w) \in \mathbf{S}_{t(a,w_p)} \ \text{and} \ \mathbf{T}_{s(w)}^{\succ t(a,w_p)} = \emptyset
\\
\lambda_{w}\prod\limits_{t(a',w') \in \mathbf{T}_{s(w)}^{\succ t(a,w_p)} }\left ( 1-\rho _{t(a',w')} \right ), \text{if} \ s(w) \in \mathbf{S}_{t(a,w_p)} \ \text{and} \ \mathbf{T}_{s(w)}^{\succ t(a,w_p)} \ne \emptyset
\end{cases}.
\end{equation}
\end{figure*}

The probability ${\rho _{t\left( {a,w} \right)}}$ of of having at least one taker in state $t(a,w)$ at any moment depends on the average arrival rates of unpaired passengers and pairing opportunities for taker-state $t(a,w)$, i.e., $\lambda_{t(a,w)}$ and $\eta_{t(a,w)}$, and the expected time that every taker is available for pairing in its taker-state, i.e., $\bar\tau_{t(a,w)}$. 
So $\rho_{t(a,w)}$ follows
\begin{equation}\label{eq:rho_t(a,w)}
\rho_{t(a,w)}
\begin{cases}
\lambda_{t(a,w)}\left [ 1-\text{exp} \left( - \eta_{t(a,w)}\overline{\tau }_{t(a,w)}\right)  \right ], \text{if} \ \eta_{t(a,w)} > 0
\\ 
\lambda_{t(a,w)}\overline{\tau }_{t(a,w)}, \text{if} \ \eta_{t(a,w)} = 0
\end{cases}.
\end{equation}

Passengers between OD pair $w\in W$ turn into takers in state $t(a,w)$ only if they are not paired in previous states, so the average arrival rate $\lambda_{t(a,w)}$ of unpaired takers for state $t(a,w)$  depends on the pairing probability in the preceding seeker- and taker-states of $t(a,w)$:
\begin{equation}\label{eq:lambda_t(a,w)}
\lambda_{t(a^{n},w)} =
\begin{cases}
\lambda_w(1-p_{s(w)}),n=0
\\ 
\lambda_{t(a^{n},w)} \left( 1-p_{t(a^{n-1},w)} \right), 1 \le n \le \left | A_w \right |
\end{cases}.
\end{equation}

Eqs. (\ref{eq:p_sw})-(\ref{eq:lambda_t(a,w)}) describe the complex interactions among the passengers' pairing probabilities in different states, i.e., $\boldsymbol{p_s}=(p_s(w),s(w)\in \mathbf{S})$, $\boldsymbol{p_t}=(p_{t(a,w)},t(a,w)\in \mathbf{T})$, the probability of passenger existence in every taker-state at any moment, i.e., $\boldsymbol{\rho}=(\rho_{t(a,w)},t(a,w)\in \mathbf{T})$, the arrival rate of pairing opportunities for each taker-state, i.e., $\boldsymbol{\eta}=(\eta_{t(a,w)}^{s(w')},t(a,w)\in \mathbf{T}_{s(w')},s(w')\in \mathbf{S})$, and the arrival rate of unpaired passengers for each taker state, i.e., $\boldsymbol{\lambda_t}=\lambda_{t(a,w)},t(a,w)\in \mathbf{T}$. To obtain the value of the $\boldsymbol{p_s}$, $\boldsymbol{p_t}$, $\boldsymbol{\rho}$, $\boldsymbol{\eta}$ and $\boldsymbol{\lambda_t}$, we have to solve Eqs.  (\ref{eq:p_sw})-(\ref{eq:lambda_t(a,w)}) simultaneously. As established in \cite{wang2021predicting}, the existence of solutions to the system of nonlinear equations (\ref{eq:p_sw})-(\ref{eq:lambda_t(a,w)}) is always guaranteed under mild conditions. In our experiments, Eqs. (\ref{eq:p_sw})-(\ref{eq:lambda_t(a,w)}) is solved with the simple fixed point iteration methods. 

\section{Proof of Proposition 1}\label{appendix:proof of proposition 1}

Let $p_1, p_2$ denote two waiting passengers, and $k_1,k_2$ respectively be the rounds of matching they have been kept waiting. 

In subsection \ref{subsec:forwad-looking matching strategy}, the Property 1 implies that if $\Bar{e}(p_1) =\Bar{e}(p_2), l(v,{p_1}) = l(v,{p_2})$, and $k_1>k_2$, then $\Bar{e}(p_1) [\frac{\Bar{e}(p_1)}{\Bar{e}(p_1)+l(v,p_1)}]\alpha^{k_1} > \Bar{e}(p_2) [\frac{\Bar{e}(p_2)}{\Bar{e}(p_2)+l(v,p_2)}]\alpha^{k_2}$. 
Apparently, this condition holds if and only if $\alpha > 1$.

The Property 2 implies that if ${k_1} = {k_2}$ and $\Bar{e}(p_1) > \Bar{e}(p_2), l(v,{p_1}) = l(v,{p_2})$, then we have \\
$\Bar{e}(p_1) [\frac{\Bar{e_1}(p)}{\Bar{e}(p_1)+l(v,p_1)}]\alpha^{k_1} > \Bar{e}(p_2) [\frac{\Bar{e}(p_2)}{\Bar{e}(p_2)+l(v,p_2)}]\alpha^{k_2}$. \\
It is easy to see that this condition holds as long as $\alpha > 0$.

The Property 3 implies that there exists a $k^* \in [0,K/\Delta t]$, such that the utility of waiting at the origin is no greater than the utility of matching with any feasible vehicle $v\in \vV_{0,p}\cup\vV_{1,p}$, i.e.,
\begin{equation}\label{eq:cond for property 3}
  \bar e(0,p) \leq \min\begin{Bmatrix}
\Bar{e}(p) [\frac{\Bar{e}(p)}{\Bar{e}(p)+l(v,p)}]\alpha^{k^*},  &v\in \vV_0\\ 
e(v_{p'},p) [\frac{e(v_{p'},p)}{e(v_{p'},p)+l(v,p)}]\alpha^{k^*},  &v\in \vV_1\\
\end{Bmatrix}.
\end{equation}

According to Eq. (\ref{eq:e(0,p)}), when $k \rightarrow K/\Delta t$, $\bar e(0,p) \rightarrow 0$, both $\Bar{e}(p) [\frac{\Bar{e}(p)}{\Bar{e}(p)+l(v,p)}]\alpha^k > 0 $ and $e(v_{p'},p) [\frac{e(v_{p'},p)}{e(v_{p'},p)+l(v,p)}]\alpha^k > 0$ always hold, which means $ \exists k^* \in [0,K/\Delta t]$, s.t. \\
 \[\bar e(0,p) = \Bar{e}(p) [\frac{\Bar{e}(p)}{\Bar{e}(p)+l(v,p)}]\alpha^{k^*}, \]\\
or $$ \bar e(0,p) = e(v_{p'},p) [\frac{e(v_{p'},p)}{e(v_{p'},p)+l(v,p)}]\alpha^{k^*}.$$

This completes the proof.

\bibliographystyle{apalike} 
\bibliography{literature}

\end{document}